%% file: main.tex
\documentclass{article}
\usepackage[a4paper, total={16cm, 21cm}]{geometry}
\usepackage[utf8]{inputenc}
\usepackage[english]{babel}
\usepackage{csquotes}
\usepackage[title]{appendix}
\usepackage{subfiles, hyperref, parskip, stackengine, enumitem, diagbox}
\usepackage{amsthm, amsmath, amsfonts, bm, amssymb}
\usepackage{graphicx, tikz}
\usepackage[dvipsnames]{xcolor}
\usepackage[position=bottom]{subfig}
\usepackage{array, longtable, multirow, rotating, subcaption}
\usepackage{epsfig, svg}
\usepackage{packages/quantum, packages/symbols}

\theoremstyle{plain}
\newtheorem{theorem}{Theorem}
\newtheorem{lemma}[theorem]{Lemma}

\theoremstyle{definition}
\newtheorem{definition}[theorem]{Definition}

\newcommand\wa{\raisebox{.1ex}{$\vartriangleright$}}
\newcommand\bd{\raisebox{-.125ex}{$\blacksquare$}}
\newcommand\ba{\raisebox{.1ex}{$\blacktriangleright$}}

\makeatletter
\def\thm@space@setup{
  \thm@preskip=10pt plus 2pt minus 2pt
  \thm@postskip=10pt plus 2pt minus 2pt
}
\makeatother

\usepackage[style=phys, biblabel=brackets, sorting=none]{biblatex}
\bibliography{refs.bib}

\usepackage[blocks]{authblk}

\title{Instability of the undecidable behavior of the spectral gap in 1D}
\usepackage{orcidlink}
\author[1,2]{Laura Castilla-Castellano\,\orcidlink{0009-0007-6934-0758} \thanks{{\tt laurca04@ucm.es}}}
\author[3]{Angelo Lucia\,\orcidlink{0000-0003-1709-1220}\thanks{{\tt angelo.lucia@polimi.it}}}
   
\affil[1]{
   Departamento de Análisis Matemático y Matemática Aplicada, 
   Universidad Complutense de Madrid, Facultad de Ciencias Matemáticas, 28040 Madrid, Spain
}

\affil[2]{
   Instituto de Ciencias Matemáticas, 28049 Madrid, Spain
}

\affil[3]{
    Dipartimento di Matematica, Politecnico di Milano, 20133 Milano, Italy
}

\date{\today}

\begin{document}

\maketitle

\begin{abstract}
    The problem of determining the existence of a spectral gap in a lattice quantum spin system was previously shown to be undecidable for one \cite{Bausch_2020} or more dimensions \cite{Cubitt2015, CPW22}. In this work, we focus on the 1-dimensional result, showing that the constructed family with undecidable behavior is extremely sensitive to perturbations. In particular, for any $\varepsilon > 0$, there exists a 1-local, rank 1, perturbation with norm $O(\varepsilon)$, such that the spectral gap problem for the family in \cite{Bausch_2020} now becomes decidable.
\end{abstract}

\tableofcontents

\subfile{1_Introduction}
\subfile{2_Previous}
\subfile{3_Perturbed}
\subfile{4_Discussion}

\section*{Acknowledgements}
We thank David Pérez-García, for insightful discussions and valuable feedback.

L.C. acknowledges financial support from grants PID2020-113523GB-I00, PID2023-146758NB-I00 and CEX2023-001347-S, funded by MICIU/AEI/10.13039/501100011033, as well as grants PRE2021-098747, funded by MICIU/AEI/10.13039/501100011033 and “ESF+”, and TEC-2024/COM-84-QUITEMAD-CM, funded by Comunidad de Madrid.

A.L. acknowledges support from the Italian Ministry of University and Research (MUR), through ``Programma per Giovani Ricercatori Rita Levi Montalcini'', the grant ``Dipartimento di Eccellenza 2023-2027'' of Dipartimento di Matematica, Politecnico di Milano, as well as the National Group of Mathematical Physics (GNFM) of the Italian Institute for High Mathematics (INdAM).

\printbibliography

\clearpage
\appendix
\subfile{A_Appendix}

\end{document}

%% file: 1_Introduction.tex
\section{Introduction}\label{sec:introduction}

For quantum many-body systems, the spectral gap is arguably one of the most physically significant properties. Whether a system is gapped or gapless has far-reaching consequences: it governs the nature of quantum phase transitions \cite{Sondhi_1997}, controls the decay of correlations and the structure of entanglement in ground states \cite{Hastings_2006, Hastings_2007_entanglement}, determines the efficiency of classical and quantum simulation algorithms \cite{arad2013arealawsubexponentialalgorithm}, and underlies the classification of topological phases of matter \cite{Wen_2017}. Nonetheless, despite its central importance, the question of whether a given local Hamiltonian is gapped or gapless was first proven to be undecidable for quantum spin systems in two or more dimensions \cite{CPW22, Cubitt2015}.

A follow-up question was to study the case in one dimensional systems, often regarded as the most tractable class of quantum many-body systems. Ground states of gapped 1-dimensional systems satisfy an entanglement area law \cite{Hastings_2007, arad2013arealawsubexponentialalgorithm}, and can therefore be efficiently approximated by matrix product states. The density matrix renormalization group (DMRG) algorithm performs extremely well in 1-dimensional \cite{PhysRevLett.69.2863}, and there exist polynomial-time algorithms for computing ground state properties of all gapped 1-dimensional Hamiltonians \cite{landau2015polynomial}. Furthermore, 1-dimensional quantum systems cannot exhibit thermal phase transitions \cite{PhysRevB.10.2900} or intrinsic topological order \cite{PhysRevLett.94.140601}, and for the simplest class of such systems (frustration-free nearest-neighbor qubit chains) the spectral gap problem has been completely solved \cite{10.1063/1.4922508}. Even classical analogues such as tiling and satisfiability problems become tractable in 1D, in contrast to the NP-hardness and undecidability that arise in two or more dimensions \cite{berger1966undecidability, 10.1145/800157.805047}. However, despite its apparent simplification, the question was again proven to remain undecidable even in the 1-dimensional setting \cite{Bausch_2020}.

Nevertheless, the fine-tuning of the construction raises a natural and physically important question: how robust are these undecidability results to perturbations of the local interaction terms? Stability is a foundational concept in physics, as the usefulness of theoretical descriptions of a physical system rest on the premise that its predictions do not collapse under the small imperfections of real settings: noise, manufacturing tolerances, finite measurement precision, and the truncation of infinitely many degrees of freedom to a tractable model.

Whether undecidable behavior in quantum spin systems persist or simplifies under perturbations of the local interactions is then a natural question. The authors of \cite{Bausch_2020} raise it explicitly in their discussion, and make two key observations. On the one hand, they note that perturbations of the Feynman-Kitaev history state terms encoding the computation are not robust \cite{BC18}: such perturbations tend to disrupt the computational encoding, localizing the ground state, and predictably destroying the desired undecidable behavior (which does not mean, however, that it immediately becomes decidable). On the other hand, they argue that small perturbations of the classical diagonal coupling terms leave the construction intact, since these only shift the relevant penalty and bonus terms by an amount proportional to the perturbation strength, too small to alter the thermodynamic behavior.

In this work, we show that the undecidable family of \cite{Bausch_2020} is in fact extremely fragile with respect to a class of perturbations. Even if the perturbation we study here is also classical and diagonal, it is an additional 1-local term, which disrupts the fine-tuning energy coupling in the construction, breaking the relationship between the spectral gap and the undecidable problem encoded in the interactions. Specifically, we prove that for any $\varepsilon > 0$, adding this 1-local, rank 1 interaction of norm $O(\varepsilon)$, which is also classical (i.e., diagonal in the computational basis), is sufficient to render the spectral gap problem for the resulting models decidable. That is, the spectral gap of the perturbed family can be determined by a finite computation whose complexity depends only on $\varepsilon$ and the system constants.

The paper is organized as follows. The remainder of Section \ref{sec:introduction} reviews the necessary computational and quantum many-body background, as well as the structure of the previous results \cite{Bausch_2020, Cubitt2015, CPW22}. Section \ref{sec:1d_undecidable} provides a self-contained analysis of how the 1-dimensional family of Hamiltonians with undecidable behavior is constructed \cite{Bausch_2020}. Section \ref{sec:1d_perturbed} introduces the perturbation, analyzes its effect on the previously defined family with undecidable behavior, and proves decidability of the spectral gap for the perturbed family.

\subsection{Computational notions}

We recall the standard definitions of undecidability and the Halting problem.

\begin{definition}\label{def:undecidable_problem}
    An \textit{undecidable problem} is a decision problem (one with a yes/no output) that has been proven to have no general algorithm for its resolution.
\end{definition}

\begin{definition}\label{def:halting_problem}
    Given an arbitrary pair $(P, I)$ describing a program and an input, the \textit{Halting problem} is the decision problem of determining if a program $P$ on input $I$ will finish in a finite amount of steps or will continue to run forever.
\end{definition}

In 1936, Alan Turing showed that the Halting problem is an undecidable problem \cite{turing1936computable}. One common strategy to prove that other problems are also undecidable is to encode the Halting problem in them. This encoding is usually done by representing the programs $P$ as Turing Machines (TM), and in particular by the use of a so called Universal Turing Machine (UTM), namely a Turing Machine can replicate the behavior of any arbitrary TM by reading both instructions $P$ and input $I$ (with the pair usually being encoded in a single integer $n \in \mathbb{N}$). One then needs to connect the desired problem with the halting behavior of this UTM: this is precisely what it is done in for the spectral gap undecidability results \cite{Bausch_2020,Cubitt2015, CPW22}

For the purposes of this work, we take the same construction around Turing machines that appears in \cite{Bausch_2020}, without modification, but for the sake of the reader, we recall in Appendix \ref{sec:turing-machines} the definitions of both a classical and a quantum Turing Machine, as taken from \cite{BV93} (Section 3), and their universal versions.

\subsection{Quantum systems framework}
We will be using the standard framework for describing finite quantum spin chains. In a chain of $N$ particles, we associate to each site $i$ a Hilbert space $\mathcal{H}^{(i)} \simeq \mathbb{C}^d$, and Hamiltonians will be given by translation invariant 1-body and 2-body interactions as follows.

\begin{definition}\label{def:1d_hamiltonian_definition}
Given a on-site interaction $h^{(1)} \in \mathcal{B}(\mathcal{H})$ and a nearest-neighbor interaction $h^{(2)} \in \mathcal{B}(\mathcal{H} \otimes \mathcal{H})$, we define a 1-dimensional, translationally invariant, nearest-neighbor Hamiltonian on a chain of $N$ particles as 
\begin{equation}
        H_N =  \sum_{i=1}^N h^{(1)}_i + \sum_{i=1}^{N-1} h^{(2)}_{i,i+1},
\end{equation}
where $h^{(1)}_i \in \mathcal{B}(\mathcal{H}^{(i})$ are translations of $h^{(1)}$ while $h^{(2)}_{i,j}  \in \mathcal{B}(\mathcal{H}^{(i)} \otimes \mathcal{H}^{(i+1)})$ are translations of $h^{(2)}$.
\end{definition}

\begin{definition}\label{def:gs_energy_definitions}
Consider a family of $1$-dimensional Hamiltonians $\{H_N: N \in \mathbb{N}\}$ indexed by the size $N \in \mathbb{N}$ of the chain. We have the following definitions:
    \begin{enumerate}
        \item The \textit{spectral gap} of $H_N$ is $\Delta(H_N):= \lambda_1(H_N)-\lambda_0(H_N)$.
        \item The family $\{H_N: N \in \mathbb{N}\}$ is \textit{gapped}, if there is a constant $\gamma>0$ and a system size $N_0$ such that for all $N>N_0$, $\lambda_0(H_N)$ is non-degenerate and $\Delta(H_N) \geq \gamma$. In this case, we say that the spectral gap is at least $\gamma$.
        \item The family $\{H_N: N \in \mathbb{N}\}$ is \textit{gapless}, if there is a constant $c>0$ such that for all $\varepsilon>0$ there is an $N_0\in\mathbb{N}$ so that for all $N>N_0$, any point in $[\lambda_0(H_N),\lambda_0(H_N)+c]$ is within distance $\varepsilon$ from $\spec(H_N)$.
        \item The \textit{spectral gap decision problem} asks whether, given a description of the interactions $h^{(1)}$ and $h^{(2)}$, the corresponding family of 1-dimensional Hamiltonians $\{H_N: N \in \mathbb{N}\}$ is gapped or gapless, under the promise that exactly one of the two cases holds.
    \end{enumerate}
\end{definition}

\subsection{Overview of spectral gap undecidability results}\label{subsec:original_overview}
Both the 2-dimensional \cite{CPW22} and 1-dimensional \cite{Bausch_2020} result share the same structure:

\begin{enumerate}
    \item Encode the evolution of a UTM in the ground state of a local Hamiltonian. For details about encoding computations in ground states of Hamiltonians, see the work by Gottesman and Irani \cite{GI10}.
    
    \item Create a ground state energy difference between halting and non-halting instances. This can be done by either penalizing the appearance of the halting state, so that the ground state energy is 0 if the halting state is not achieved and positive otherwise (as it is the case in \cite{CPW22,Cubitt2015}), or vice-versa by penalizing the non-halting behavior, so that the ground state energy is 0 in the non-halting instances and positive otherwise (as it is the case in \cite{Bausch_2020}). In either case, this energy difference can only be made of order $\Theta(1/T^2)$, where $T$ is the number of steps taken by the UTM machine. In both cases, the same issue exists: the energy difference between halting and non-halting instances will vanish on larger system sizes. Therefore, this poses the problem of finding ways to amplify this energy difference.
    
    \item Amplify the energy difference. In both constructions, the energy penalty from the previous step is amplified by allowing for multiple instances of the Turing Machine computation to happen in parallel. In the 2-dimensional result, this can be achieved by exploiting the aperiodic structure of the 2-dimensional Robinson tiling \cite{robinson1971undecidability}. However, for the 1-dimensional case, this approach does not work directly, as there are no aperiodic tilings in 1D. The solution to this problem introduced in \cite{Bausch_2020} is to create a ``Marker Hamiltonian'',  which induces partitions in the chain, forming consecutive segments separated by a boundary state $\ket{\bd}$, and whose length and period are related to the halting behavior. In case of halting instances, the energy landscape favors the appearance of multiple copies of finite segments of the optimal length (long enough to see the halting behavior, but not longer than that in order to avoid reducing the energy penalty), thus producing the desired amplification of the $\Theta(1/T^2)$ energy penalty.  In case of non-halting instances, the optimal segment length is unbounded, resulting in a unique segment covering the full chain. 

    \item Coupling with gapped and gapless instances. By designing a coupling between the Hamiltonian encoding the UTM computation (with the amplification  of the energy penalty) with fixed gapped and gapless Hamiltonians, once can control the behavior of the spectral gap in terms of the halting behavior of the encoded problem, thus making the spectral gap uncomputable for this family of models.

\end{enumerate}

We summarize the main result from \cite{Bausch_2020} as follows:
\begin{theorem}[\cite{Bausch_2020}, Theorems 1 and 25]\label{th:undecidable}
Fix an Universal Turing Machine (UTM). There exist (explicitly constructible) local interactions $h^{(1)}(\eta)$ and nearest-neighbor interactions $h^{(2)}(\eta)$, parametrized by an integer $\eta$, such that the family of Hamiltonians $\{H_N(\eta)\}_N$ defined on a spin chain with $N$ sites by \[H^u_N(\eta) = \sum_{i=1}^Nh_i^{(1)}(\eta) + \sum_{i=1}^{N-1}h_{i,i+1}^{(2)}(\eta)\] satisfies the following:
    \begin{enumerate}
        \item If the UTM halts on input $\eta$, then $\{H^u_N(\eta)\}_N$ is gapless.
        \item If the UTM does not halt on input $\eta$, then $\{H^u_N(\eta)\}_N$ has a gap of at least $1$ for all $N \in \mathbb{N}$.
    \end{enumerate}
\end{theorem}

\subsection{Main result}\label{subsec:main_result}
In our work, we consider the family of 1-dimensional Hamiltonians defined in Theorem \ref{th:undecidable} from \cite{Bausch_2020}, and we perturb it with a 1-local term of arbitrarily small strength $\varepsilon > 0$. This perturbation will alter the mechanism behind the marking process, disabling the energy amplification step, and thus reducing the problem of separating gapped and gapless instances to a decidable one. Our result can be summarized as follows:

\begin{theorem}\label{th:main_result}
Let $\varepsilon > 0$ be arbitrarily small, and consider the local interactions $h^{(1)}(\eta)$ and nearest-neighbor interactions $h^{(2)}(\eta)$ from Theorem \ref{th:undecidable}. Then there exists a 1-local, rank 1, interaction $p^{(1)}_\varepsilon = \varepsilon\ket{\bd}\!\!\bra{\bd}$, with norm $\norm{p^{(1)}_\varepsilon} = O(\varepsilon)$, such that the spectral gap problem for local interactions $h^{(1),\varepsilon}(\eta) = h^{(1)}(\eta) +  p^{(1)}_\varepsilon$ and $h^{(2)}(\eta)$ is decidable.
In fact, denoting
\[H^{\varepsilon}_N(\eta) = \sum_{i=1}^N ( h_i^{(1)}(\eta) + p^{(1)}_\varepsilon) + \sum_{i=1}^{N-1}h_{i,i+1}^{(2)}(\eta),\]
and by $\lambda^N_{min}$ the ground state energy of $H^{\varepsilon}_N(\eta)$,
there exists a critical size $w_\varepsilon$, independent of $\eta$ and $N$, such that exactly one of the following two conditions holds: 
\begin{enumerate}
    \item $\lambda^N_{min} \geq 1$ for all $N \in [1, \dots, w_\varepsilon]$.
    \item $\lambda^N_{min} < 0$ for some $N \in [1, \dots, w_\varepsilon]$.
\end{enumerate}
Then we have that in the first case, $\{H_N^{\varepsilon}(\eta)\}_{N}$ is gapped, otherwise it is gapless. 
Since checking whether $\lambda^N_{min} <0$ for some $N \in [1, \dots, w_\varepsilon]$ is a finite program, the spectral gap problem for the perturbed interactions is decidable.
\end{theorem}

%% file: 2_Previous.tex
\section{Analysis of the 1D undecidable Hamiltonian}\label{sec:1d_undecidable}
We start by giving a summarized outline of the work done in \cite{Bausch_2020}, as it is fundamental in order to see how the undecidability result breaks if slightly perturbed. We start by studying the first ingredient, the computational Hamiltonian, which will remain unperturbed.

\subsection{Computational Hamiltonian}\label{subsec:computational_hamiltonian}
As usual in these constructions, the manner of encoding a computation in a ground state is based in \cite{GI10}, by constructing a history state Hamiltonian, denoted in \cite{Bausch_2020} as $H_{TM}(M, \phi(\eta))$. Here, $M$ is a classical universal Turing machine, and $\phi(\eta)$ encodes the input parameter $\eta \in \mathbb{N}$ with $|\eta|$ binary digits as a binary fraction, interleaved by 1s: \[\phi(\eta) = 0.\eta_11\eta_21\dots\eta_{|\eta|}00\]

This particular way of encoding the input is done to face an issue previous to the computation: to be able to deterministically write any possible (binary) input, and check that it was done correctly. This question was solved in Section 3 in the 2-dimensional result \cite{CPW22}. The procedure used in the 1-dimensional version of the result is a modification of it, as the problem to circumvent is the same: what happens if the procedure is truncated?

In the 2-dimensional construction, one could have the help of knowing exactly the size of the segment one is working on. However, the segment size here is unknown, so they tail some extra steps in order to check if the binary expansion was done or truncated, arriving at a way of locally checking it: in the former case, the overlap of the least significant qubit with state $\ket{-}$ is $0$, and in the latter, at least $\mu = 2^{|\phi|}$.

By inflicting a penalty to a head state over $\ket{-}$, one can avoid encountering incorrect inputs in the ground state configurations. However, this only ensures a penalty of energy $\mu/T^3$. In order to avoid it being negligible, the marking Hamiltonian that will be later constructed should also be scaled by $\mu$ as well, in order to have both computational and marking Hamiltonians in comparable orders of magnitude.

Once this preliminary step is done, one can start to build the desired 1-dimensional computational Hamiltonian. In order to encode the input properly, the procedure based on quantum phase estimation (\cite{Bausch_2020}, section V) runs first. Then, M, that will be an universal TM, uses that encoded input and runs the computation on it. We summarize the behavior of this first Hamiltonian in the following definition:

\begin{lemma}[\cite{CPW22} Theorem 33, combined with \cite{Bausch_2020} Lemma 15]\label{th:htm}
The Hamiltonian $H_{TM}(M, \phi(\eta))$ as defined in \cite{Bausch_2020}, sections III.B and V, is 1-dimensional, translationally invariant and nearest-neighbor, and behaves as follows:
\begin{enumerate}
    \item A Quantum Turing machine performs a quantum phase estimation procedure on a single-qubit unitary, that encoding input $\eta \in \mathbb{N}$ as $\phi(\eta) = 0.\eta_11\eta_21\dots\eta_{|\eta|}00$.
    \item The classical universal Turing machine M uses this binary expansion form of the input to perform a computation on it.
\end{enumerate}
\end{lemma}

\subsection{Marking the segments}\label{subsec:marking}
As mentioned in Section \ref{subsec:original_overview}, the idea behind the Marking Hamiltonian is to solve the energy amplification issue. To understand how it is constructed, in this section, we summarize the results of \cite{Bausch_2020}, Section IV.

The goal of the Marking Hamiltonian is first to induce some partitions in the chain, turning it into different segments divided by boundary state $\ket{\bd}$. Start by defining special boundary state $\ket\bd$. The goal is to use it to fragment the chain in different partitions or segments:

\[ \ket{\bd \dots \bd \dots \bd}\]

and then, create a ``walk'' Hamiltonian, that mimics the behavior of a clock in each of them:

\[
\begin{aligned}
&\ket{\bd \wa \wa \wa \dots \wa \bd} \\
&\ket{\bd \ba \wa \wa \dots \wa \bd} \\
&\ket{\bd \ba \ba \wa \dots \wa \bd} \\
&\phantom{\bd \ba \ba \wa }\vdots \\
&\ket{\bd \ba \ba \ba \dots \ba \bd}
\end{aligned}
\]

For that purpose, the Hamiltonian $H = H_{walk} + P + P'$ is created, where:

\begin{itemize}
    \item $H_{walk}$ describes the allowed transitions: $\ket{\ba\wa\wa} \rightarrow \ket{\ba \ba \wa}$, and $\ket{\ba\wa\bd} \rightarrow \ket{\ba \ba \bd}$.
    \item $P$ encodes the penalties $\ket{\bd\bd}$, $\ket{\wa\ba}$ and $\ket{\bd\wa}$, needed to ensure the correct evolution of the clock as well.
    \item $P'$ collects a bonus of $-4$ to every appearance of $\ket{\bd}$, and a $+2$ to them appearing next to any symbol.
\end{itemize}

With this description, any properly bounded chain as according to $H$ will have a ground state energy of $-4$, in the form of many $\ket{\bd}$-bounded segments. Due to the construction, the state between boundaries lies within the span of states $\{ \ket{1} = \ket{\ba\wa\wa\dots\wa\wa}, \ket{2} = \ket{\ba\ba\wa\dots\wa\wa}, \dots, \ket{w-1} = \ket{\ba\ba\ba\dots\ba\wa}, \ket{w} = \ket{\ba\ba\ba\dots\ba\ba}\}$, where $w$ denotes the length of the segment (without boundaries).

The authors then add another bonus term, $B$, that gives $-1$ energy to states where the arrow has reached the right boundary, which is just one state between the set mentioned above. Taking this into account, a careful analysis of the energy of the $w$-segment (segment of size $w$) bounded by $\ket{\bd}$, as studied in \cite[section IV.D]{Bausch_2020}, yields that its ground state energy lies in the interval $\bigg(-\dfrac{1}{2} - \dfrac{1}{2^w}, \;\; -\dfrac{1}{2} - \dfrac{1}{4^w}\bigg)$. This is done by noting that the transition rules in $H_{walk}$ result in a Hamiltonian that is precisely the Laplacian of the graph determined by said transition rules, where each state is a vertex, and the allowed transition, an edge between them. Additionally, they find that this is the unique negative eigenvalue, thus being a Hamiltonian with a gap of at least $1/2$ (\cite{Bausch_2020}, Lemma 6 and Corollary 9).

By adding an additional penalty $P''$ of $1/2$ to each apparition of $\ket{\bd}$, one can shift the $-1/2$ appearing in the previous result. Then, the minimal energy configuration of $H' = H + P'' + B$ over the chain is guaranteed to be a properly bounded configuration of multiple $w_i$-segments, with a minimal energy that is:
\begin{itemize}
    \item The sum of all segment energies.
    \item A plus of $1/2$ for each boundary $\ket{\bd}$.
    \item The $-4$ bonus that resulted from being in a  properly bounded configuration as defined by $H$.
\end{itemize}

By canceling each $-1/2$ of the segment with its right boundary penalty of $1/2$, this then results in a total energy of $-7/2 + \sum_i \lambda_i$, where the $-7/2$ is the result of the $-4$ plus the extra $1/2$ that remains from the leftmost boundary, and that can be removed with a constant energy shift in the system. Thus, the final energy given by the marker, $\lambda_f$ is bounded as $-\sum_i 1/2^{w_i} \leq \lambda_f \leq -\sum_i 1/4^{w_i}$.

Additionally, one has to guarantee that the bonus given by this marking is strictly smaller than the computational penalties, in order to make the final construction to be dominated by the computational behavior. To do so, the authors in \cite{Bausch_2020} explain in Section IV.E how the unary counter described by $\ket{\wa}$ and $\ket{\ba}$ can be changed to allow more complicated 2-local transition rules, with a consequent increase in the local dimension. Then, instead of a lineal dependence on the segment length $w$, one gets a dependence on $f(w)$: the new length of the path graph evolution on a segment of length $w$. Then, the final marking Hamiltonian, $H_f$, has a final ground state energy of: $-\sum_i 1/2^{f(w_i)} \leq \lambda_f \leq -\sum_i 1/4^{f(w_i)}$.

In their Remark 12, they state that one can have $f(w) = (d-c_2)^w$, for local dimension $d$ and constant $c_2$, giving then a double exponential decay of the energy of the $w$-segments, which is enough to guarantee the dominance of the computational behavior over the formation of markings. In particular, the authors used $f(w)=2^w$, guaranteeing a double exponential decay in the marking bonuses. We summarize the behavior of the marking Hamiltonian over the chain in the following theorem:

\begin{theorem}[\cite{Bausch_2020}, Theorem 10]\label{th:marking_hamiltonian}
The Hamiltonian $H_f = H_{walk} + P+ P' + P'' + B$ as defined in \cite{Bausch_2020}, section IV, is 1-dimensional, translationally invariant and nearest-neighbor, and:
\begin{itemize}
    \item Its ground state configuration is a properly bounded configuration of multiple $w_i$-segments, separated by $\ket{\bd}$.
    \item Its ground state energy is
    \begin{equation}\label{eq:segment_contribution}
        -\sum_i 1/2^{f(w_i)} \leq \lambda_f \leq -\sum_i 1/4^{f(w_i)}
    \end{equation}
    \item It has a gap of size $\geq 1/2$.
\end{itemize}
\end{theorem}

\subsection{Combined Hamiltonian}
Once the structure of the chain is defined, the next step is to combine it with the computation. As explained in \ref{subsec:computational_hamiltonian}, in order to avoid incorrect inputs in the ground state, a penalty must be inflicted to a head state over $\ket{-}$, and we also need to  and differentiate the halting and non-halting behaviors, by penalizing the non-halting case. Collect these two penalties in the term $P''' = \ket{q_?;-}\bra{q_?;-} \sum_i\ket{h_i\bd}\bra{h_i\bd}$, where $\{h_i\}$ is the set of head states that are going to be penalized if found next to the boundary (i.e.: all of them), and $\ket{q_?;-}$ is the state, as per \cite{Bausch_2020}, Section V, that indicates that the QPE procedure was not properly finished.

By combining the marking Hamiltonian $H_f$ (properly scaled by $\mu$), the previous computational Hamiltonian $H_{TM}$, and the penalties of $P'''$, the final Hamiltonian renders the following behavior over a single segment, as collected in Theorem 20 of \cite{Bausch_2020}, and as seen in Figure \ref{fig:original_segment_energy}.

\begin{theorem}[\cite{Bausch_2020}, Theorem 20]\label{th:theorem20}
The Hamiltonian $H_w = \mu H_f + H_{TM} + P'''$, over a segment of length $w$, presents the following energy behavior:
\begin{itemize}
    \item If the machine is non halting, the energy will always be positive, and decreasing with segment length, due to the penalty being $O(1/T^3)$.
    \item If the machine halts, then the behavior up to halting size is as before, but once it is big enough to halt, the energy will turn negative. However, due to the decreasing nature of the marking bonus, this will be negative but increasing. Then, the minimal energy is achieved when the segment length is exactly the one needed to halt, and no extra space.
\end{itemize}
\end{theorem}
\begin{proof}

The computational Hamiltonian encounters penalties in three out of four cases: (i) when the QPE is truncated, (ii) when the machine is non halting, (iii) or when it is halting, but the size needed for doing so is greater than $w$. A history state Hamiltonian encoding a computation of length $T$ that picks up at least one energy penalty, has ground state energy of order $\Theta(1/T^2)$ \cite{BC18}. If it halts in size $w$, the energy is simply $\lambda_{TM} = 0$. If not, in all three, a safe asymptotic bound is $\lambda_{TM} < 1/T^3$. 

Additionally, the actual runtime $T=T(w)$ of the TM on segment size $w$ is at most the Poincaré recurrence bound $T_{max}(w) = Q \times w \times A^w$, where the constants $Q$ and $A$ are the number of internal states and symbols on tape, respectively. This bound grows exponentially in $A^w$, so $T^3_{max}(w) < A^{3w}$. As the bonus given by the marker, $\lambda_f < 0$, is never greater than $\mu/2^{f(w)}$ (Theorem \ref{th:marking_hamiltonian}), by choosing an appropriate $f(w)$ such that
\begin{equation}\label{eq:marking_energy_bound}
\dfrac{1}{T^3} = \dfrac{1}{T^3(w)} > \dfrac{1}{T_{max}^3(w)} > \dfrac{1}{A^{3w}} > \dfrac{1}{2^{f(w)}} > \dfrac{\mu}{2^{f(w)}} \geq - \lambda_f
\end{equation}

we can conclude that the marking bonus is $|\lambda_f| < 1/T^3$. This can be done if $2^{f(w)} > A^{3w}$, and the authors take, in particular, $f(w)=2^w$.

The different terms of the Hamiltonian commute, so the total energy is the sum of each part, $\lambda = \lambda_f + \lambda_{TM}$. Thus:
    \begin{itemize}
        \item If the TM is non halting, then: \[ \lambda_{TM} + \lambda_f \geq \dfrac{1}{T^3} + \lambda_f > 0,\] due to $1/T^3$ being a lower bound for the computational energy, and Equation \ref{eq:marking_energy_bound}.
        \item If it is halting, but without enough space to do so, same argument as above applies. Else, $\lambda_{TM} = 0$, and $\lambda = \lambda_f < 0$, of order $-\Theta(1/2^{f(w)})$.
    \end{itemize}
\end{proof}

\begin{figure}[hbtp]
\centering
    \includegraphics[width=0.8\textwidth]{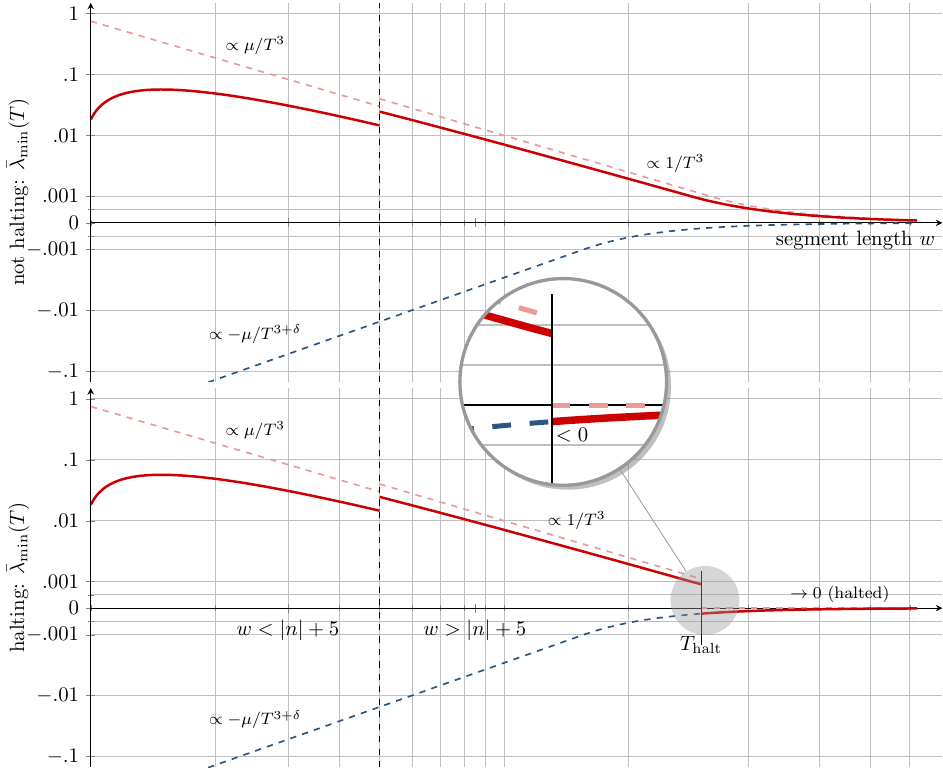}
    \caption{(Taken from \cite[Section III]{Bausch_2020}) Energy contribution $\bar\lambda_\text{min}(T)$ from a single segment of length $w$ of the marker and TM Hamiltonian $\mu H_f + H_{TM}$ shown in red, where $T$ is the runtime of the encoded computation. The dashed red and blue lines are bounds on the contribution of both $H_{TM}$ and $\mu H_f$, respectively.}
     \label{fig:original_segment_energy}
\end{figure}

We finish this section by summarizing the main technical theorem, Theorem 22 in \cite{Bausch_2020}, which explains the behavior of $H_w$ over a chain of size $N$, instead of over a single segment, which we denote by $H_N$.

\begin{theorem}[\cite{Bausch_2020}, Theorem 22]\label{th:combined_on_chain}
For any Turing machine $M$ and input $\eta \in \mathbb{N}$, we can explicitly construct a sequence of 1D, translationally invariant, nearest-neighbor Hamiltonians $H_N(\eta, M)$ on the Hilbert space $(\mathbb{C}^d)^{\otimes N}$ with the property that either:
    \begin{enumerate}
        \item $M(\eta)$ does not halt, and $\lambda(H_N) \geq 0$ for all $N \in \mathbb{N}$.
        \item $M(\eta)$ halts, and
        \[ \lambda(H_N) \begin{cases} 
                            < - \lfloor N/w_{halt} \rfloor \Omega(1/T_{halt}^{3}) & \text{if $N>w_{halt}$} \\
                            \geq 0 & \text{if $N\le w_{halt}$}
                        \end{cases}
        \]
        where $T_{halt}$ is the time needed for $M(\eta)$ to halt, and $w_{halt}$ is the length of the tape accessed during the computation.
    \end{enumerate}
\end{theorem}
\begin{proof}
Due to the fact that $H_{TM}$ and $P'''$ both commute with $H_f$, the minimum energy of $H_N$ is the sum of the energies of all the induced segments. There are two cases, depending on the halting behavior:
\begin{enumerate}
    \item The computation does not halt. In that case, every time it reaches a boundary, it will be penalized. Additionally, the contribution of a segment in that case falls monotonically with segment length. Therefore, the lowest energy possible is achieved by having only one segment over the chain, which has a non-negative energy.
    \item The computation halts after $T_{halt}$ time, having consumed $w_{halt}$ tape.
    \begin{itemize}
        \item If $N < w_{halt}$, the same argument as above holds, resulting in a positive energy.
        \item Else, as seen in \ref{th:theorem20}, the best energy bonus over a segment is achieved at exactly the size needed for halting, contributing $-\Omega(1/T^3)$, as per Equation \ref{eq:marking_energy_bound}. It is beneficial then to have as much segments of size $w_{halt}$ as possible, namely, $\lfloor N/w_{halt} \rfloor$ of them, summing the maximum possible bonuses, contributing a total of $-\lfloor N/w_{halt} \rfloor\Omega(1/T^3)$. One could have a remaining right-most segment that is shorter than the halting size, but this adds a single constant energy penalty, over $O(N)$ segments that induce a bonus, so the asymptotic bound remains.
    \end{itemize}
\end{enumerate}
\end{proof}

\subsection{Undecidability results}

The final part of the result in \cite{Bausch_2020} is to first shift the energy of $H_N$, and then make use of a a gapped and a gapless Hamiltonian. When combined, the undecidable behavior of the spectral gap appears. We finish this section by recalling these final steps in the construction.

\begin{lemma}[\cite{Bausch_2020}, Lemma 23]\label{lem:lemma23}
By adding at most two-local identity terms, we can shift the energy of $H_N$ from Theorem \ref{th:combined_on_chain} such that
\[
    \lambda_0(H_N)    \begin{cases}
                            \ge 1 & \text{if the TM does not halt,} \\
                            \longrightarrow-\infty & \text{in the halting case.}
                        \end{cases}
\]
\end{lemma}
\begin{proof}
Employ \textcite{GI10}'s boundary trick, as done when constructing $P'$ in \ref{th:marking_hamiltonian}, which relies on the fact that there is $N$ one-local but only $N-1$ two-local terms.
\end{proof}

\begin{lemma}\label{lem:haux}
The Hamiltonians $H_0 = \sum_{i=1}^N \ket{1}\bra{1}_i$ and $H_d = \sum_{i=1}^N \dfrac{1}{2} (\sigma^x_i\sigma^x_{i+1} +\sigma^y_i\sigma^y_{i+1})$ over $(\mathbb{C}^2)^{\otimes{N}}$ have the following properties:
    \begin{itemize}
        \item $H_0$ is 1-local, translationally-invariant, diagonal in the computational basis, has unique zero-energy ground state $\ket{00\dots0}$, and all other $\lambda \in \text{spec}(H_0)$ satisfy $\lambda \geq 1$.
        \item $H_d$ is 2-local, translationally-invariant, and has continuous spectrum in $[0, \infty)$ in the thermodynamic limit.
    \end{itemize}
\end{lemma}
\begin{proof}
The first one is the trivial Hamiltonian, with spectrum $\{0,1,2,\dots, N\}$. The second one is the 1-dimensional critical XY model, found and solved in \cite{lieb1961two}.
\end{proof}

\begin{theorem}[\cite{Bausch_2020}, Theorems 1 and 25]\label{th:undecidable-2}
Take an Universal Turing Machine (UTM). There exist (explicitly constructible) local interactions $h^{(1)}(\eta)$ and nearest-neighbor interactions $h^{(2)}(\eta)$, parametrized by an integer $\eta$, such that the family of Hamiltonians $\{H_N(\eta)\}_N$ defined on a spin chain with $N$ sites and local dimension $d$ by \[H^u_N(\eta) = \sum_{i=1}^Nh_i^{(1)}(\eta) + \sum_{i=1}^{N-1}h_{i,i+1}^{(2)}(\eta)\] satisfies the following:
    \begin{enumerate}
        \item If the UTM halts on input $\eta$, then $\{H^u_N(\eta)\}$ is gapless.
        \item If the UTM does not halt on input $\eta$, then $\{H^u_N(\eta)\}$ has a gap of at least $1$ for all $N \in \mathbb{N}$.
    \end{enumerate}
\end{theorem}
\begin{proof}
Take $H_N$ from Theorem \ref{th:combined_on_chain}, after having shifted its energy according to Lemma \ref{lem:lemma23}, and $H_0$, $H_d$ from Lemma \ref{lem:haux}, acting on Hilbert spaces $\mathcal{H}_C$, $\mathcal{H}_d$ and $\mathcal{H}_0$ respectively. Construct the Hamiltonian \[H^u_N = H_N \otimes \mathbbm{1}_2 \oplus 0_3 + \mathbbm{1}_1\otimes H_d \oplus 0_3 + 0_{1,2}\oplus H_0 + H_g\] on $\mathcal{H} := (\mathcal{H}_C \otimes \mathcal{H}_d) \oplus H_0$, where $H_g := \sum_{i=1}^{N} (\mathbbm{1}_{1,2}^{(i)} \otimes \mathbbm{1}_{3}^{(i + 1)} + \mathbbm{1}_{3}^{(i)} \otimes \mathbbm{1}_{1,2}^{(i + 1)})$ ensures that any state with support on both $\mathcal{H}_c \otimes \mathcal{H}_d$ and $\mathcal{H}_0$ has an energy penalty.
\end{proof}

%% file: 3_Perturbed.tex
\section{Decidability of the perturbed family}\label{sec:1d_perturbed}
Having reviewed the procedure used to construct the family with undecidable behavior, in this section we explain how, for every $\varepsilon > 0$, there exists a 1-local interaction of norm $O(\varepsilon)$ term that renders the problem decidable for the perturbed family.

\subsection{Perturbing the marker Hamiltonian}\label{sec:perturbation}
When constructing the marking Hamiltonian, the penalty of $1/2$ per $\ket{\bd}$ is tuned in order to cancel the the $-1/2$ that appears in the energy analysis of the interior part of the segment. The idea is to alter this fine tuning, by any given arbitrarily small perturbation of $\varepsilon >0$, and study how this translates to the bigger picture.

Therefore, we start by changing the term $P''$ to $P''_\varepsilon = 1/2+\varepsilon$. Then, the new marking Hamiltonian, $H_f^\varepsilon$, behaves, in  contrast to Theorem \ref{th:marking_hamiltonian}, as follows:

\begin{theorem}\label{th:marking_epsilon_hamiltonian}
The Hamiltonian $H^\varepsilon_f = H_{walk} + P+ P' + P''_\varepsilon + B$ is 1-dimensional, translationally invariant and nearest-neighbor, and:
\begin{enumerate}
    \item Its ground state configuration is a properly bounded configuration of multiple segments, separated by $\ket{\bd}$.
    \item If the formation has segments of different sizes $W = \{w_1, \dots, w_W\}$, then its ground state energy is \[|W|\varepsilon-\sum_{w_i \in W} 1/2^{f(w_i)} \leq \sum_i \lambda_i + 7/2 - \varepsilon \leq |W|\varepsilon -\sum_{w_i \in W} 1/4^{f(w_i)}\]
    \item It has a gap of size $\geq 1/2$.
\end{enumerate}
\end{theorem}
\begin{proof}
By Theorem 10 of \cite{Bausch_2020}, having no double boundaries and both ends properly marked by $\bd$ is energetically favorable:
\begin{itemize}
    \item Having a double boundary $\ket{\bd\bd}$ is penalized with $+2$ by $P$. Deleting one of the boundaries results in a configuration with lower energy, due to removing the contribution from $P$. However, a longer segment has been created, but as its contribution is comprised in $(-1/2-1/2^{f(w)}, -1/2-1/4^{f(w)})$, the bonus is consequently smaller. However, the difference is no more than $1/2$, and therefore is still better to avoid the $+2$ penalty from $P$.
    \item Having no bounded ends is not energetically favorable either. Given $P'$, moving an inner boundary to an edge gives an extra $-2$ energy, because the $\ket{\bd}$ has now one less neighbor. Again, the decrease of the bonus due to the larger segment is no more than $1/2$, so the bounded ends configuration is more favorable.
\end{itemize}
Therefore, the configuration of the chain is described by a properly bounded chain into segments with no double boundaries. As every segment is terminated by a boundary, the penalty of $1/2+\varepsilon$ offsets the $-1/2$ of its contribution, but a $\varepsilon$ per segment remains. Additionally, the properly bounded configuration given by $P'$ has a $-4$ energy contribution. In the original construction, the final Hamiltonian was already lifted by a global shift of $7/2$, in order to compensate these constants (see Section \ref{subsec:marking}). However, now an additional $\varepsilon$ remains from the leftmost boundary, so we include an additional global shift of $\varepsilon$ in order to have our energy discussion centered around $0$ as well. The gap claim follows from the spectral gap of the original $H_f$, due to Lemma $5$ and Corollary $9$ in \cite{Bausch_2020}.
\end{proof}

\subsection{Energy analysis}
We learned, in section \ref{sec:perturbation}, how the modified marker Hamiltonian $H^\varepsilon_f$ will still induce partitions in the chain. However, how will the energy behave when combined with the computational Hamiltonian $H_{TM}$?

Let's start by contrasting with Theorem \ref{th:theorem20}, about what happens in a single segment, of size $w$.
\begin{theorem}\label{th:perturbed_energy}
The Hamiltonian $H^{\varepsilon}_w = \mu H^\varepsilon_f + H_{TM} + P'''$, over a segment of size $w$, has a ground state energy of:
 \[ \lambda_\varepsilon(w) = \lambda_{TM}(w) + \mu\varepsilon - \delta(w), \;\;\;\; \text{with} \; \begin{cases}
                            \bullet \;\; \delta(w) \in \bigg[\dfrac{\mu}{4^{f(w)}}, \dfrac{\mu}{2^{f(w)}}\bigg] \\[9pt]
                            \bullet \;\; |\delta(w)| < 1/T^3 \\[5pt]
                            \bullet \;\; \lambda_{TM} \in \Theta(1/T^2)
                        \end{cases}
\]
and presents the following energy behavior:
\begin{itemize}
    \item If the TM does not halt, halts but not before $w$, or has a truncated input, the ground state energy is positive.
    \item If the TM halts at size $w_{halt} < w$, but $\delta(w_{halt})< \mu\varepsilon$, then the ground state energy is positive too.
    \item If the TM halts at size $w_{halt} \leq w$, and $\delta(w_{halt}) \geq \mu\varepsilon$, only then the ground state energy is negative.
\end{itemize}
\end{theorem}
\begin{proof}
We know, by Theorem \ref{th:marking_epsilon_hamiltonian}, over a single segment, the minimal energy is the sum of both the computational and the marking energy, that can be written as above. Then, again, we focus on the different cases in halting and non halting instances:
\begin{enumerate}
    \item The TM does not halt. Then, for all sizes $w$, it will have a $\lambda_{TM}(w) \in \Theta(1/T^2)$. We need to split now into further cases for the marking energy.
        \begin{enumerate}
            \item If $\mu\varepsilon > \delta(w)$ for all sizes $w\in\mathbb{N}$, then, the energy coming from the marking Hamiltonian, $\mu\varepsilon - \delta(w)$ is always positive as well, rendering the final energy positive as well.
            \item If we are not in the previous case, this means that $\delta(w) \geq \mu\varepsilon > 0$ for some $w$. As $\mu$ and $\varepsilon$ are two positive constants, and $\delta(w)$ is monotone decreasing, this means that the condition only happens for the first $k$ sizes. That is, $w\in [1, \dots, k]$. Then, \[0 \geq \mu\varepsilon - \delta(w) > -\delta(w) > - \mu/2^{f(w)},\] which means that the marking contribution is negative. However, as explained in Theorem \ref{th:theorem20}, the falloff exponent $f(w)$ is chosen accordingly in order to not alter the positive sign given by the computational penalty, with the goal of making the computational behavior lead the construction. Therefore, the energy will remain positive.

            Lastly, for all sizes greater than $k$, we are back in case 1(a), and $\lambda_\varepsilon > 0$ as well.
        \end{enumerate}
    \item The TM halts at size $w_{halt}$. Then, $\lambda_{TM} = 0$ for all $w \geq w_{halt}$. Again, we need to check the balance with the marking contribution.
        \begin{enumerate}
            \item If $w_{halt} > w$, the behavior in the chain is the same as being non halting, so the energy analysis is the same as in case 1. 
            \item If $w_{halt} \leq w$, we have to check the same sign distinction:
            \begin{enumerate}
                \item If $\mu\varepsilon > \delta(w)$ for all sizes, then again, the marking energy is always positive, and thus $\lambda_\varepsilon = \lambda_{TM} + \lambda_f > 0$, as $\lambda_{TM} \geq 0$.
                \item Let again be $[1,\dots, k]$ the sizes where the marking energy would be negative, as before. For sizes smaller than $w_{halt}$, we can go back to case 1.b, as we are not halting in there.
                
                The remaining situation is when we have sizes in the interval that are big enough to halt. Here, now the marking bonuses are negative, and monotone increasing, but never surpassing $0$, being the best bonus the one achieved by a segment of size exactly $w = w_{halt}$. Thus, for every $w \geq w_{halt}$, we have that $\lambda_\varepsilon = \lambda_f < 0$.
            \end{enumerate}
        \end{enumerate}
    \item As explained in Theorem \ref{th:theorem20}, if the input was truncated, the behavior will be the same as a non-halting instance over that segment size, so energy analysis corresponds to case 1 as well.
    \end{enumerate}
\end{proof}
Theorem \ref{th:perturbed_energy} still has a distinction over halting and non halting instances. However, it is dependent on a critical size: where the marking energy contribution $\lambda_f$ switches from being negative to being positive. As this critical size is only dependent on constants $\mu$ and $\varepsilon$, and not on the halting problem, it is computable. We define it as follows:

\begin{definition}
For all positive constants $\mu, \varepsilon$, there exists another computable constant $w_\varepsilon \in \mathbb{N}$ such that $\mu\varepsilon - \delta(w) > 0$ for all $w \geq w_\varepsilon$. We call $w_\varepsilon$ the \textit{critical size}.
\end{definition}

With this, we can construct the final technical theorem of this section, which informally tells us that we only need to check up to the value of the energy in this critical size, in order to be able to determine the value in the thermodynamic limit.

\begin{theorem}\label{th:decidable_gs}
For any Turing Machine $M$ and input $\eta \in \mathbb{N}$, we can construct a $1$D, translationally invariant, nearest-neighbor family of Hamiltonians $\{H_N\}_{N=1}^{w_\varepsilon}$ over chain sizes $\{1, 2, \dots, w_\varepsilon\}$, such that their ground state energies $\lambda^N_0 \not = 0$. Then:
    \begin{enumerate}
        \item If $\lambda^N_0 > 0$ for all $N \in [1, \dots, w_\varepsilon]$, then we have $\lim_{N \to \infty}\lambda^N_0 > 0$.
        \item If $\lambda^N_0 < 0$, for some $N \in [1, \dots, w_\varepsilon]$, then we have $\lim_{N \to \infty}\lambda^N_0 < 0$.
    \end{enumerate}
\end{theorem}
\begin{proof}
Theorem \ref{th:perturbed_energy} shows us that halting later than the critical size still yields a positive energy, as if it were non halting. Therefore, we only need to be concerned with what happens up to $w_\varepsilon$. If all $\lambda^N_0$ are positive, this means that the TM never showed a halting behavior up to that point, so by the explained above, the energy will remain positive in the thermodynamic limit. On the other hand, if the TM is halting with halting time $w_{halt} < w_\varepsilon$, then at least at $w_{halt}$, the energy is negative. We do not need to be concerned with the truncation of the input, that happens if $w < |\eta|+5$. In the original construction \cite{Bausch_2020}, this was a problem because it could alter the later halting or non-halting behavior. However, in the perturbed scheme, after $w_\varepsilon$, these two instances are not distinguishable anymore, and thus the important threshold now is only dependent on $\varepsilon$.

The need of checking \textit{up to} $w_\varepsilon$ instead of \textit{just} $w_\varepsilon$ comes from the fact that the finite size behavior is dependent on the specific constants. The maximum bonus is given at size $w_\varepsilon$, but not all segments are proportional to it. The asymptotic behavior, as explained in Theorem \ref{th:combined_on_chain}, is making as many $w_{halt}$-segments as possible, where having at most a single right-most segment that does not overthrown the negative sign. However, over finite size, as in the original construction, it could be the case that the right most penalty of is over a (shorter) size, and therefore is not guaranteed to leave the energy negative. Thus the ground state configuration could be a different arrangement of segment of other sizes, and consequently of unknown sign. 
\end{proof}

\subsection{Decidable spectral gap}\label{sec:decidable}
The final original family of Hamiltonians $\{H_N^u\}$ of \ref{th:undecidable}, combines the construction with both a gapped and gapless Hamiltonian, as shown in Figure \ref{fig:gapped_vs_gapless}. We finish by showing in Theorem \ref{th:final_result} how the perturbation results in the combined final behavior being now decidable. This finishes proving the main result, as summarized in Section \ref{subsec:main_result}, Theorem \ref{th:main_result}.

\begin{figure}[hbtp]
\centering
  \begin{tabular}{cc}
    \includegraphics[width=0.4\columnwidth]{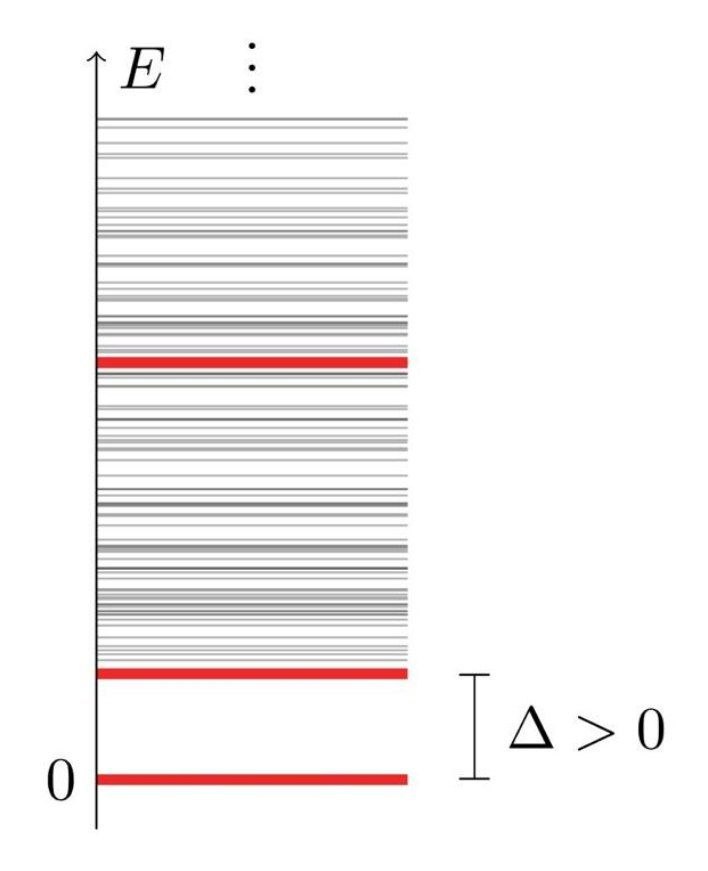} & \includegraphics[width=0.28\columnwidth]{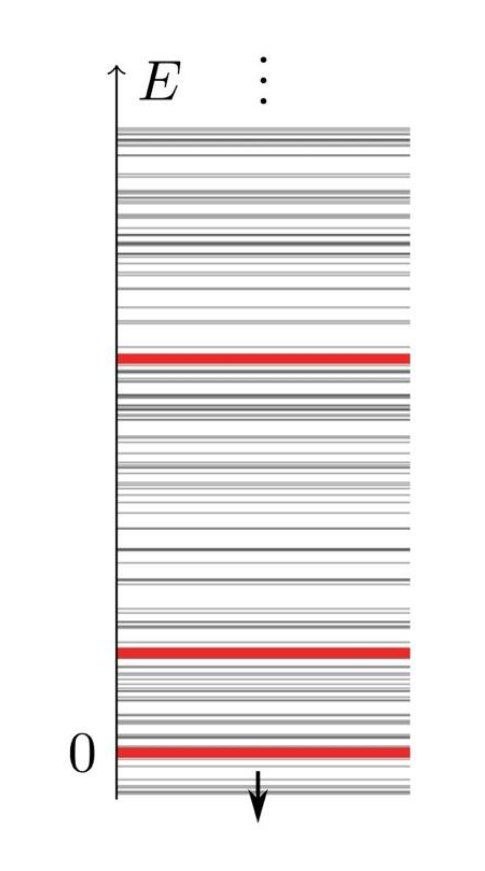}
  \end{tabular}
  \caption{(Taken from \cite{Bausch_2020}) Final spectrum of $H_N^{\varepsilon}$. The red lines represent the spectrum of $H_0$, and the gray ones, from $H_d$. On the left, we have the gapped situation, where the gap of $H_0$ shows. On the right, the gapless case.}
  \label{fig:gapped_vs_gapless}
\end{figure}

\begin{theorem}\label{th:final_result}
Take any $\varepsilon > 0$. The family of Hamiltonians $\{H_N^{\varepsilon}\}$, constructed as in Theorem \ref{th:undecidable}, but with the $\varepsilon$-perturbation in term $P''_\varepsilon$, no longer shows undecidable behavior. In fact, given critical size $w_\varepsilon$:
    \begin{enumerate}
        \item If $H_N^\varepsilon$ has $\lambda^N_{min} \geq 1$ for all $N \in [1, \dots, w_\varepsilon]$, then the family is gapped in the thermodynamic limit.
        \item If $H_N^\varepsilon$ has $\lambda^N_{min} < 0$ for some $N \in [1, \dots, w_\varepsilon]$, then the family is gapless in the thermodynamic limit.
    \end{enumerate}
\end{theorem}
\begin{proof}
Combine the family of Hamiltonians $\{H_N\}_{N=1}^{w_\varepsilon}$ constructed in Theorem \ref{th:decidable_gs}, apply the shift of Lemma \ref{lem:lemma23} and combine with $H_0$ and $H_d$ in Lemma \ref{lem:haux}, as in the final step of \cite{Bausch_2020}. The coupling with the dense spectrum Hamiltonian $H_d$ occurs at the lowest eigenvalue, so if the first eigenvalue of $H_N$ is negative, then the resulting Hamiltonian will grow to show the gapless spectrum. On the contrary, if this coupling occurs at $\lambda_{min}^N \geq 1$, the the gap with $H_0$, of ground state energy $0$, shows instead (see Figure \ref{fig:gapped_vs_gapless}). However, due to Theorem \ref{th:decidable_gs}, we can finitely check which situation holds.
\end{proof}

%% file: 4_Discussion.tex
\section{Discussion}
We have shown that the undecidable spectral gap family of \cite{Bausch_2020} is actually very sensitive to perturbations: for any $\varepsilon > 0$, a 1-local, rank 1 perturbation $p_\varepsilon = \varepsilon\ket{\bd}\bra{\bd}$ of the boundary symbol, switches the behavior of the spectral gap from undecidable to decidable, with a complexity that depends only on $\varepsilon$ and other constants.

The mechanism that makes this perturbation work is based on perturbing an exact energy cancellation within the Marker Hamiltonian, on which the energy amplification step (and therefore the encoding of the Halting problem) depends. This result sharpens the stability discussion already present in \cite{Bausch_2020}. The authors observe that perturbations of the classical diagonal coupling terms are safe, reasoning that the relevant penalty and bonus terms are of order one, so a perturbation of strength $\varepsilon$ only shifts them by a proportionally small amount. The perturbation in this work is also classical and diagonal, but its acting regime has an energy scale not of order one, but exponentially small in the segment length. This shows that depending on the precise nature of the perturbation, the undecidable behavior can drastically switch.

This finding is consistent with the state of the art on spectral gap stability. Existing rigorous stability theorems require strong structural assumptions, such as Local Topological Quantum Order (LTQO) or frustration-freeness \cite{Bravyi_2010, Michalakis_2013}. The construction of \cite{Bausch_2020} is neither frustration-free nor LTQO, and indeed, outside this regime, this work shows that even an elementary single-site classical perturbations can have a drastic qualitative effect on the low-energy spectrum.

The undecidability of the spectral gap has by now been extended along several directions beyond the original 1D and 2D constructions \cite{Bausch_2020, Cubitt2015, CPW22}: phase diagrams themselves can be uncomputable \cite{Bausch_2021}, approximating critical points of quantum phase transitions was shown to be of high complexity \cite{Watson_2021}, uncomputability results were also extended to renormalization group flows \cite{Watson_2022}, and imposing additional constraints such as symmetry does not in general remove undecidability \cite{CCL24}. These constructions naturally raise the question of how physically robust the underlying mechanisms are: is the corresponding undecidable or uncomputable behavior stable when faced with small perturbations of the interactions, or, as we find for the 1-dimensional spectral gap, could it drastically change?

Additionally, a line of recent work revolves around another related notion of the gap: the bulk gap. In a recent work \cite{Xu_2026}, it was shown that the bulk spectral gap, defined via the KMS condition on infinite-volume states rather than as a limit of finite-volume open-boundary spectra, is semi-decidable, via a convergent hierarchy of certified upper bounds obtained from semidefinite programming. This notion of gap is distinct, though not contradictory, from the finite-volume notion underlying \cite{Bausch_2020} and the present work, as they formalize different physical questions. The finite-volume notion is a property of a specific family of Hamiltonians with open boundary conditions, and the undecidable construction relies essentially on boundary-induced structure (the marker symbol, the segmentation of the chain, and the Gottesman-Irani boundary trick all exploit the mismatch between one and two-body terms that exists only because the chain has open ends). The bulk gap, by contrast, is defined intrinsically on the translation-invariant infinite-volume dynamics, with no reference to any boundary condition. It is known that finite-volume spectra with boundary conditions need not reflect the bulk gap: a system can be gapped in the bulk yet exhibit gapless edge states once a boundary is introduced \cite{Bachmann_2015}. In a related fashion, a hierarchy of certified lower bounds for the spectral gap of frustration-free systems was developed \cite{Rai_2026}, complementing the complete decidability result for frustration-free nearest-neighbor qubit chains \cite{10.1063/1.4922508}. Together, these developments suggest that undecidable behavior can be considerably sensitive, although the extent to which this phenomenon is generic remains an open question.

Finally, for the 2D result \cite{Cubitt2015, CPW22}, whether an analogous instability holds is unclear, since energy amplification there is driven by the very rigid, self-similar structure of the Robinson tiling \cite{robinson1971undecidability} rather than a single fine-tuned boundary term, and it is not thought that any comparably simple perturbation could disrupt it. Another question is the high local dimension of these constructions: is there a local-dimension threshold below which the spectral gap problem is necessarily decidable? A positive result in this direction is the frustration-free, nearest-neighbor qubit chain case of \cite{10.1063/1.4922508}, but which is specific to 1D.

%% file: A_Appendix.tex
\section{Turing machines}
\label{sec:turing-machines}

\begin{definition}\label{def:turing_machine}
  A \textit{(deterministic) Turing Machine} (TM) is defined by a triplet $(\Sigma, Q, \delta)$ where $\Sigma$ is a finite alphabet with an identified blank symbol $\#$, $Q$ is a finite set of states with an identified initial state $q_0$ and final state $q_f\not = q_0$, and $\delta$ is a transition function
      \begin{equation}\label{eq:classical_transition_function}
        \delta: Q\times \Sigma \rightarrow \Sigma\times Q\times \{L,R\}.
      \end{equation}
  The TM has a two-way infinite tape of cells indexed by $\mathbb{Z}$ and a single read/write tape head that moves along the tape. A configuration of the TM is a complete description of the contents of the tape, the location of the tape head and the state $q\in Q$ of the finite control. At any time, only a finite number of the tape cells may contain non-blank symbols.
  
  For any configuration $c$ of the TM, the successor configuration $c'$ is defined by applying the transition function to the current state and the symbol scanned by the head, replacing them by those specified in the transition function and moving the head left (L) or right (R) according to $\delta$.

  By convention, the initial configuration satisfies the following conditions: the head is in cell $0$, called the \textit{starting cell}, and the machine is in state $q_0$. We say that an initial configuration has input $x\in (\Sigma\setminus\#)^*$ if $x$ is written on the tape in positions $0,1,2,\dots$ and all other tape cells are blank. The TM halts on input $x$ if it eventually enters the final state $q_f$. The number of steps a TM takes to halt on input $x$ is its running time on input $x$. If a TM halts, then its output is the string in $\Sigma^*$ consisting of those tape contents from the leftmost non-blank symbol to the rightmost non-blank symbol, or the empty string if the entire tape is blank. A TM is called \textit{reversible} if each configuration has at most one predecessor.
\end{definition}

\begin{definition}
    An \textit{Universal Turing Machine} (UTM) is a Turing Machine capable of simulating any other Turing Machine. That is, $UTM(TM, n) = TM(n)$ for every Turing Machine TM and input $n$.
\end{definition}

\begin{definition}
  Call $\tilde{\mathbb{C}}$ to the set of $\alpha \in \mathbb{C}$ such that there is a deterministic algorithm that computes the real and imaginary parts of $\alpha$ to within $2^{-n}$ in time polynomial in $n$. A \textit{Quantum Turing Machine} (QTM) is defined by a triplet $(\Sigma, Q, \delta)$ where $\Sigma$ is a finite alphabet with an identified blank symbol $\#$, $Q$ is a finite set of states with an identified initial state $q_0$ and final state $q_f\neq q_0$, and a quantum transition function
  \begin{equation}\label{eq:quantum_transition_function}
    \delta: Q\times \Sigma \rightarrow \tilde{\mathbb{C}}^{\Sigma\times Q\times \{L,R\}}.
  \end{equation}
  The QTM has a two-way infinite tape of cells indexed by $\mathbb{Z}$ and a single read/write tape head that moves along the tape. We define configurations, initial configurations and final configurations exactly as for deterministic TMs.

  Let $\mathcal{S}$ be the inner-product space of finite complex linear combinations of configurations of the QTM, which we call $M$, with the Euclidean norm. We call each element $\phi\in \mathcal{S}$ a superposition of $M$. The QTM $M$ defines a linear operator $U_M: \mathcal{S}\rightarrow\mathcal{S}$, called the time evolution operator of $M$, as follows: if $M$ starts in configuration $c$ with current state $p$ and scanned symbol $\sigma$, then after one step $M$ will be in superposition of configurations $\psi=\sum_i\alpha_ic_i$, where each nonzero $\alpha_i$ corresponds to the amplitude $\delta(p,\sigma,\tau,q,d)$ of $\ket{\tau}\ket{q}\ket{d}$ in the transition $\delta(p,\sigma)$ and $c_i$ is the new configuration obtained by writing $\tau$, changing the internal state to $q$ and moving the head in the direction of $d$. Extending this map to the entire $\mathcal{S}$ through linearity gives the linear time evolution operator $U_M$.
\end{definition}

\begin{definition}
    Equivalently to the classical case, we say that a \textit{Universal Quantum Turing Machine} (UQTM) is a Quantum Turing Machine capable of simulating any other Quantum Turing Machine. That is, $UQTM(QMT, n) = QTM(n)$ for every Quantum Turing Machine QTM and input $n$.
\end{definition}